\documentclass[runningheads]{llncs}

\usepackage[T1]{fontenc}
\usepackage{graphicx}
\usepackage{amssymb}
\usepackage{complexity}
\usepackage{mathtools}
\usepackage{enumerate}

\newcommand{\dsparse}{$d$-sparse }

\newcommand{\dnd}{$(d,n,d)$ }
\newcommand{\sqn}{\sqrt{n}}
\usepackage{todonotes}
\usepackage{mathtools}
\DeclarePairedDelimiter{\ceil}{\lceil}{\rceil}
\usepackage{graphicx}
\usepackage[ruled,vlined]{algorithm2e}
\usepackage{todonotes}
\usepackage{mathtools}

\usepackage{graphicx}

\newif\ifanonymous
\anonymousfalse          

\usepackage[hidelinks]{hyperref}
\usepackage{cleveref}
\ifanonymous
  \hypersetup{
    pdfauthor={},      
    pdfcreator={},
    pdfproducer={}
  }
\fi

\begin{document}
\title{Matrix Multiplication in the MPC model}
\ifanonymous
%
\else
\author{Lakshya Joshi \inst{1}\orcidID{0009-0004-6521-9182} \and
Arya Deshmukh\inst{2}\orcidID{0009-0000-3413-0493
} \and
Atharv Chhabra\inst{1}\orcidID{0009-0007-0498-7067} \and
Chetan Gupta\inst{3}\orcidID{0000-0002-0727-160X}}
\authorrunning{Joshi et al.}
%

\institute{Squarepoint Technologies, India \and
Microsoft, India
\\
\and
Indian Institute of Technology, Roorkee, India\\
}
\fi
\maketitle              
\begin{abstract}
In this paper, we present algorithms to solve matrix multiplication problems in the MPC model. In particular, we consider the problem under various processor/memory constraints in the MPC model and prove the following results. 

\begin{enumerate}
   \item Multiplication of two rectangular matrices of size $d \times n$ and $n \times d$ ( where $d \leq n$) respectively can be done in, 
    \begin{enumerate}
        \item $O(\sqrt{d} + \log_d n)$ rounds with $n$ processors and $\Theta(d)$ memory per processor
        \item $O(\frac{d}{\sqrt{n}})$ rounds with $d$ processors and $\Theta(n)$ memory per processor.
    \end{enumerate}
    \item Multiplication of two rectangular matrices of size $n \times d$ and $d \times n$ (where $d \leq n$) respectively, with $n$ processors of $\Theta(n)$ memory per processor can be done in $O(\frac{d}{\sqrt{n}})$ rounds.
    \item The multiplication of two $d$-sparse matrices (matrices that contain at most $d$-nonzero elements in each row and in each column) with $n$ processors and $\Theta(d)$ memory per processor can be done in $O(d^{0.9})$ rounds.
\end{enumerate} 
\end{abstract}
\section{Introduction}
We are given two matrices $A$ and $B$, and our task is to compute $ C=A \times B$. Matrix multiplication has been extensively studied in the centralized setting, leading to a long line of work on improving the asymptotic complexity of sequential algorithms (see, e.g., \cite{pan1978strassen,coppersmith1987matrix,coppersmith1990matrix,vassilevska2012multiplying,legall2014powers,alman2021refined}). However, these algorithms are not efficiently parallelizable due to the inherent dense communication pattern in them. This limitation has motivated the investigation of matrix multiplication in distributed and parallel computational models, including the clique model, low-bandwidth models, and related frameworks  \cite{le2016further,Censor-HillelLT18,alg-method-congest-cliq2019,censor2021fast,gupta-2022-sparse,gupta-2024-sparse}. Building on this line of research, in this paper we explore matrix multiplication in the Massively Parallel Computation (MPC) model, a framework that captures large-scale parallel computation with restricted memory per machine. We consider the problem $C= A \times B $ in the following three settings: when matrices $A$ and $B$ are (i) rectangular matrices of dimension $d \times n$ and $n \times d$, respectively, (ii) rectangular matrices of dimension $n \times d$ and $d \times n$, respectively and, (iii) $d$- sparse matrices i.e. both are $n \times n$ size matrices such that each matrix contains at most $d$ non-zero elements in each row and each column. For completeness, we provide a brief overview of the MPC model in \cref{subsec:mpcmodel}, followed by a precise description of the computational settings in which we address these three variants of the matrix multiplication problem.

\subsection{MPC Model}
\label{subsec:mpcmodel}
In this paper, we work with the \textit{massively parallel computation} model (MPC) \cite{Howard10}. Given any problem of size (input $+$ output) $m$  words, such that $m$ is too big to store on a single computer, in the MPC model, we assume that the problem is distributed over a fully connected network of computers (or processors). The computation proceeds in synchronous rounds. Processors perform computation locally, and at the end of a round, they exchange their messages. However, unlike some other models of distributed and parallel computation, in the MPC model, we assume that each computer has a limited memory. In particular, we assume that each computer has $\Theta(m^{\delta})$ memory and the number of computers in the network is $\Theta(m^{1-\delta})$, where $0< \delta <1$. Therefore, the total memory of the system is $\Theta(m)$ sufficient to store the given problem. Since each computer has $\Theta(m^{\delta})$ memory, it is allowed to send and receive $O(m^{\delta})$ words in each round. If any processor receives more than $\Theta(m^{\delta})$ words during the course of an algorithm, the algorithm is considered failed. 

\subsection{Setup}
\label{subsec:setup}
We assume that initially, matrices $A$ and $B$ are distributed evenly among all the processors. The structure of the distribution does not matter because we can reach any desired even distribution in one round (even if we start with an uneven distribution). At the end of the algorithm the elements of the output matrix $C$ should also be evenly distributed. We analyse the problem in the following cases. 
\begin{enumerate}[(i)]
  \item $(d,n,d)$: when $A$ and $B$ are rectangular matrices size $d \times n$ and $n \times d$ respectively, where $d \leq n$.
    \item $(n,d,n):A$ and $B$ are of size $n \times d$ and $d \times n$ respectively 
  \item $d$-sparse: $A$ and $B$ are $d$-sparse $n \times n$ matrices, i.e. each row and column of matrices $A$ and $B$ contains at most $d$ non-zero elements. Also we are interested in at most $d$ nonzero entries in each row and each column in $C$
\end{enumerate}

\subsection{Our Results}

\paragraph{\((d,n,d)\):}  
In this case, both the input and output matrices have size \(O(dn)\). Accordingly, it suffices to assume that the total memory available in the system is \(\Theta(dn)\). We analyze this instance under two different processor--memory configurations:  

\begin{enumerate}[(i)]
    \item For $n$ processor each with $\Theta(d)$ memory we design an algorithm with round complexity  $O(\sqrt{d} + \log_d n)$

    \item For $d$ processor each with $\Theta(n)$ memory we design an algorithm with round complexity $(\tfrac{d}{\sqrt{n}})$ 
    rounds.
\end{enumerate}

\paragraph{\((n,d,n)\):}  
Here, the input matrices have size \(O(nd)\), while the output matrix may be as large as \(O(n^2)\), corresponding to an \(n \times n\) matrix. Therefore, we analyse this case under the configuration of \(n\) processors, each with \(\Theta(n)\) memory, and present an algorithm with round complexity $O(\tfrac{d}{\sqrt{n}})$.

\paragraph{\(d\)-sparse case:}  
Finally, we consider the case where both \(A\) and \(B\) are \(n \times n\) matrices, each with at most \(d\) nonzero entries per row and per column. Since we are interested in $d$ elements in the output matrix, the size of both input and output matrices is bounded by \(O(nd)\). Therefore, we analyse this case with $n$ processors each with $\Theta(d)$ memory. The state-of-the-art for this is a trivial algorithm that takes $O(d)$ rounds. We improve that by giving an $O(d^{0.9})$ round algorithm.  

\subsection{High-level Idea of Our Technique}
In rectangular cases, we observed that in order to reduce inter-processor communication, the best case is when all processors try to compute a square submatrix of $C$. Thus, we divide the matrices $A$ and $B$ into smaller square sub-matrices of equal size depending on the local memory of the processor, and assign each submatrix to one of the processors. At the end, each processor needs to produce the corresponding submatrix of $C$.
For the \dsparse case, we adapt the idea of Gupta et al.\cite{gupta-2022-sparse}. We first refine the analysis given in their paper and improve upon their results. Here, we would like to highlight that we improve the bound present in their paper from $O(d^{1.927})$ to $O(d^{1.925})$ for solving $ d$-sparse matrix multiplication in the low bandwidth model. We do not explicitly mention that result here because that result is about the low-bandwidth model, and in this paper, we stick to the MPC model. However, the improved bound can be directly derived from the \cref{lem:sparsemm-improved} that we prove in this paper (by appropriately setting the value of $\epsilon$ for the low bandwidth model). Then we use \cref{lem:sparsemm-improved} to prove our result for $d$-sparse matrix multiplication in the MPC model. In their paper, they show that  \dsparse matrix multiplication is equivalent to processing triangles in a $ d$-degree graph. In other words, processing triangles is equivalent to a product $a_{ij}b_{jK}$, where $a_{ij}$ and $b_{jK}$ are some elements of matrices $A$ and $B$, respectively. They show that given any $d$-degree graph, say $T$, we can divide the graph into two parts such that one part, say $T'$, is \textit{clustered} - it contains layers of subgraphs of the original graph such that in each layer we have disjoint clusters of triangles which are small and dense, and the other part, say $T'$, contains a low number of triangles. Now, in order to process all the triangles, we can first process all the triangles in the clustered part layer-by-layer using square matrix multiplication, then compute all the triangles $T'$ using brute force. We noticed that there was some suboptimality in bring down the unprocessed triangles from $T$ to $T'$. In this paper, we remove that suboptimality and then use that result as a black box to solve $d$-sparse matrix multiplication in the MPC model.

\section{Rectangular Matrix Multiplication}

In this section, we consider the case when matrices $A$ and $B$ are rectangular matrices. There can be many types of instances in rectangular matrix multiplication. In this paper, we focus on two primary cases of type $(d,n,d)$ and $(n,d,n)$ and analyse them separately in subsequent sections.

\subsection{(d,n,d) Matrix Multiplication}
\label{subsec:dnd1}
In the $(d,n,d)$ case, the size of all three matrices is bounded by $nd$; thus, it is sufficient to have a total memory of size  $O(nd)$. We analyse this case under two different memory constraints, two subcases: (i) when there are $n$ processors each with $O(d)$ memory, and (ii) when there are $d$ processors each with $O(n)$, where $d \leq n$. First, we prove the following two lemmas that will help us to analyse both these cases. 

\begin{figure}
\centering
        \includegraphics[width =\textwidth]{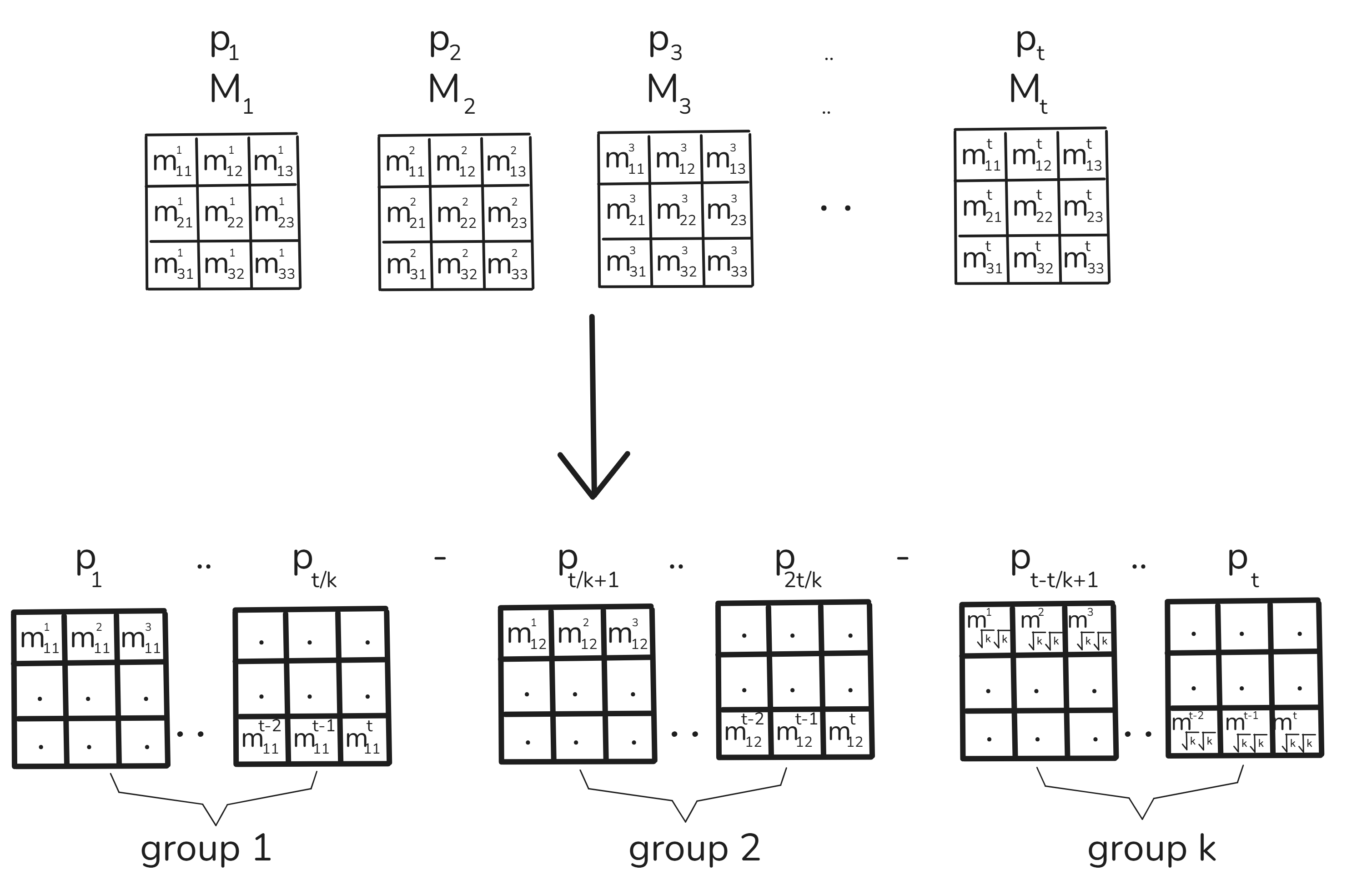}
    \caption{Upper layer represents the initial distribution of elements, and lower layer represents the distribution after the first round}
    \label{fig:mat-dist}
\end{figure}

\begin{lemma}
\label{lem:logkt}
If there are $t$ processors each containing matrices of size $\sqrt{k} \times \sqrt{k}$, and the memory of each processor is $\Theta(k)$, then the sum of these matrices can be computed in $O(\log_k t)$ rounds in the MPC model. 
\end{lemma}
\begin{proof}
    Suppose $M_1,M_2, \ldots, M_t$ are the $t$ matrices and $X = \sum M_l$. Let $m_{ij}^l$ denote the element at $i^{th}$ row and $j^{th}$ column of $M_l$. Therefore $x_{ij} = \sum m_{ij}^l$. We divide the $t$ processors into $k$ groups of $(\frac{t}{k})$ processors such that each group is responsible for computing one $x_{ij}$. Let's just focus on the first group, the same algorithm will be followed in the other groups. Let $p_1,p_2, \ldots, p_{t/k}$ be the processors in the first group. In the first round, all $t$ processor exchange messages such that $p_1$ receives elements $m_{11}^1, m_{11}^2, \ldots  m_{11}^k$; $p_2$ receives elements $m_{11}^{k+1},m_{11}^{k+2}, \ldots m_{11}^{2k}$ and so on (see \cref{fig:mat-dist} ). Now, each $p_i$ computes the sum of all the elements it receives. Notice that all the values that are required to compute $x_{11}$ are present in $t/k$ processors, such that each processor contains only one value (the sum it has calculated). Now, in the next round, processor $p_2,p_3, \ldots p_k$ will send all the intermediate values they have calculated in the previous round to $p_1$. Similarly, $p_{k+2}, p_{k+3}, \ldots p_{2k}$ will send all the intermediate values they have calculated in the previous round to $p_{k+1}$ and so on (see \cref{fig:tree} for reference). Now, notice that all the values that are required to compute $x_{11}$ are now present in the $(\frac{n}{k^2})$ processors, and each contains only one value. Repeat the previous step for $O(\log_k t)$ rounds. After $O(\log_k t)$, the final value of $x_{11}$ will be present in one processor. The same thing will happen in the other groups. Therefore, after $O(\log_k t)$ rounds processors will calculate $X$. This finishes the proof of the lemma. Notice that in the case where $t \leq k$, the algorithm will terminate in merely a single round. Therefore, in the above proof, we can assume $t>k$.
\end{proof}

\begin{figure}
\centering
        \includegraphics[width= \textwidth]{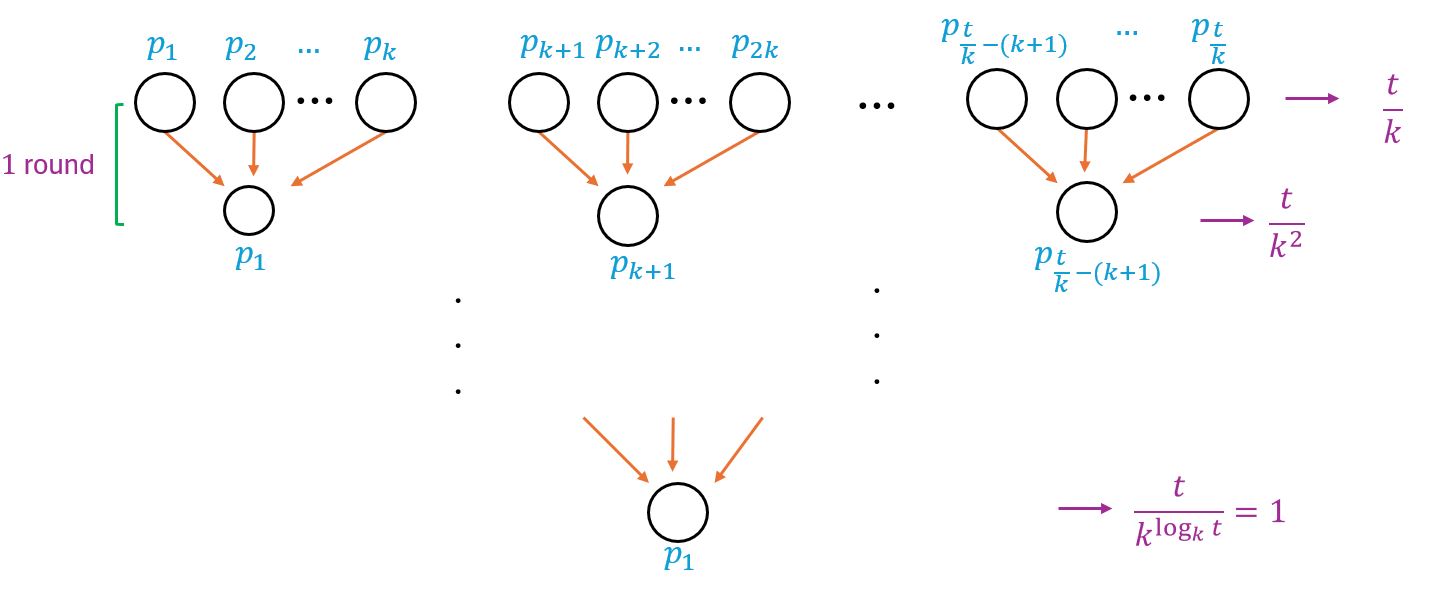}
    \caption{Processors in the first group computing the value of $x_{11}$ }
    \label{fig:tree}
\end{figure}



Now we are all set to analyse the $\dnd$ case in the aforementioned settings. First, we discuss the case with $n$ processors with $O(d)$ memory.

\subsubsection{n Processors with $O(d)$ Memory:}
Since $C$ is of size $d \times d$, let us divide $C$ into $d$ sub-matrices of size $\sqrt{d} \times \sqrt{d}$ each, denoted by $C^{ij}$, where $i,j \in [\sqrt{d}]$. We assign $(\frac{n}{d})$ processors to evaluate each sub-matrix. Also divide the matrices $A$ and $B$ into sub-matrices of size $\sqrt{d} \times \sqrt{d}$, each being denoted as $A^{iq}$ and $B^{qj}$ respectively, where $i,j \in [\sqrt{d}]$ and $ q \in [\frac{n}{\sqrt{d}}]$. We know that,
$$C^{ij} = \sum_{q}A^{iq}*B^{qj}$$

Observe that there would be a total of $n$ submatrices of the forms $A^{iq}$ and $B^{qj}$. In one round of communication, let's distribute the input among the $n$ processors such that each processor contains exactly two sub-matrices, namely, $A^{iq}$ and $B^{qj}$ for some values of $i,j,q$ such that the sub-matrices stored by any two processors are distinct. This is possible since storing $A^{iq}$ and $B^{qj}$ requires only $O(d)$ memory. Let us focus on computation on $C^{11}$, other submatrices are computed in the same way. We have assigned $\frac{n}{d}$ processors to compute it. Let $P^{11} = \{p_1, p_2, \ldots p_{\frac{n}{d}} \}$ be the set of processor assigned to $C^{11}$. We set the communication pattern as follows. In round $k$, processor $p_i$ receives submatrices corresponding to $q = (i-1)\sqrt{d}+k$. In other words, in round $1$, it receives matrices matrices $A^{1((i-1)\sqrt{d}+1)}$ and $B^{((i-1)\sqrt{d}+1)1}$, in round 2, it receives $A^{1((i-1)\sqrt{d}+2)}$ and $B^{((i-1)\sqrt{d}+2)1}$ and so on. In each round, $p_i$ computes the product of summatrices received in the current round, and adds the newly computed product to the existing matrix it computed till the previous round using simple matrix addition. It is  easy to observe that after $\sqrt{d}$ rounds, processor $p_i$ will contain matrix $\sum_{k=1}^{\sqrt{d}} A^{1((i-1)\sqrt{d}+k)} * B^{((i-1)\sqrt{d}+k)1}$. Now we have $(\frac{n}{d})$ submatrices distributed over the processors in $P^{11}$, each containing a matrix of size $\sqrt{d} \times \sqrt{d}$ such that the sum of those matrices will produce $C^{11}$. According to \cref{lem:logkt}, $C^{11}$ can be computed in $O(\log_d n)$ rounds. Similarly, the other $C^{ij}$ are also computed in parallel. Thus, matrix $C$ can be computed in $O(\sqrt{d} + \log_d n)$ rounds. The only thing missing in the proof is to show that each processor is sending messages of $O(d)$ size per round. Notice that each element of matrices $A$ and $B$ needs to be sent to at most $\sqrt{d}$ processors. Each processor initially holds $d$ elements of both $A$ and $B$. Thus, it needs to send $d\sqrt{d}$ elements to other processors, which can easily be done in $O(\sqrt{d})$ rounds, such that each processor sends $O(d)$ elements in each round. From this we get the following theorem.


\begin{theorem}
\label{thm:dnda}
 Multiplication of two rectangular matrices of size $d \times n$ and $n \times d$ respectively, with $n$ processors of $O(d)$ memory can be done $O(\sqrt{d} + \log_d n)$ rounds in the MPC model, where $d \leq n$.
\end{theorem}

\subsubsection{d processors with O(n) memory:}
    We handle this case also in a similar way to the previous one, except the division of processors and matrices is slightly different. Without loss of generality, assume that we scale the value of $d$ such that $d := (\ceil{\frac{d}{\sqn}} \sqn)$. Let us divide the matrix $C$ into $(\frac{d^2}{n})$ sub-matrices of size $\sqrt{n} \times \sqrt{n}$ each, denoted by $ C^{ij}$, where $i,j \in [\frac{d}{\sqrt{n}}]$. We assign $(\frac{n}{d})$ processors to evaluate each sub-matrix. Also divide the matrices $A$ and $B$ into sub-matrices of size $\sqrt{n} \times \sqrt{n}$, each being denoted as $A^{iq}$ and $B^{qj}$ respectively, where $i,j \in [\frac{d}{\sqrt{n}}]$ and $ q \in [\sqrt{n}]$. We know that,
$$C^{ij} = \sum_{q}A^{iq}*B^{qj}$$

Observe that there would be a total of $d$ submatrices of the form $A^{iq}$ and $B^{qj}$. In one round of communication, let's distribute the input among the $d$ processors such that each processor contains exactly two sub-matrices, namely, $A^{iq}$ and $B^{qj}$ for some values of $i,j,q$ such that the sub-matrices stored by any two processors are distinct. This is possible since storing $A^{iq}$ and $B^{qj}$ requires only $O(n)$ memory. Let us focus on computation on $C^{11}$, other submatrices are computed in the same way. We have assigned $\frac{n}{d}$ processors to compute it. Let $P^{11} = \{p_1, p_2, \ldots p_{\frac{n}{d}} \}$. In round $k$, processor $p_i$ receives submatrices corresponding to $q = (i-1)\frac{d}{\sqrt{n}}+k$ \textit{i.e.} in round 1, it receives $A^{1((i-1)\frac{d}{\sqrt{n}}+1)}$ and $B^{((i-1)\frac{d}{\sqrt{n}}+1)1}$, in round 2, it receives $A^{1((i-1)\frac{d}{\sqrt{n}}+2)}$ and $B^{((i-1)\frac{d}{\sqrt{n}}+2)1}$ and so on. It is easy to observe that $p_i$ would receive these sub-matrices from a single other processor because of the way in which the initial distribution of $A^{iq}$ and $B^{qj}$ was done among the processors. Thus, the amount of information that a processor ends up sending in any communication round would be $O(n)$. Processor $p_i$ then computes the product of these submatrices locally after each round, adding the newly computed product to the existing product that it computed in the previous round using simple matrix addition. It's easy to observe that after $\frac{d}{\sqrt{n}}$ rounds, processor $p_i$ will contain $\sum_{k=1}^{\frac{d}{\sqrt{n}}} A^{1((i-1)\frac{d}{\sqrt{n}}+k)} * B^{((i-1)\frac{d}{\sqrt{n}}+k)1}$. Now we have $(\frac{n}{d})$ submatrices distributed over the processors in $P^{ij}$, each containing a matrix of size $\sqrt{n} \times \sqrt{n}$ such that the sum of those matrices will produce $C^{11}$. According to \cref{lem:logkt}, $C^{11}$ can be computed in $O(\log_n \frac{n}{d})$ rounds. But since $\frac{n}{d} \leq n$, it can be computed in a constant number of rounds. Similarly, the other $C^{ij}$ are also computed in parallel. Thus, matrix $C$ can be computed in $O(\frac{d}{\sqrt{n}})$ rounds. 


\begin{theorem}
\label{thm:dndb}
 Multiplication of two rectangular matrices of size $d \times n$ and $n \times d$ respectively, with $d$ processors of $O(n)$ memory can be done in $O(\frac{d}{\sqrt{n}})$ rounds ounds in the MPC model, where $d \leq n$.
\end{theorem}


\subsection{(n,d,n) Matrix Multiplicaton}
In this case, we assume that we have  $n$ processors with $O(n)$ memory each. Notice that in this case the matrices $A$ and $B$ can be stored in $O(nd)$ memory; however, the resulting $C$ matrix is a square matrix of size $n \times n$. Thus, we limit the total memory of the system to $O(n^2)$. We divide the matrix $C$ into $n$ sub-matrices of size $\sqrt{n} \times \sqrt{n}$ each, denoted by $C^{ij}$ (similar to \cref{subsec:dnd1}) where $i,j \in [\sqrt{n}]$. Assign the computation of $C^{ij}$ to processor $p^{ij}$. Without loss of generality, assume that we scale the value of $d$ such that $d := (\ceil{\frac{d}{\sqn}} \sqn). $ Divide the matrices $A$ and $B$ into sub-matrices of size $\sqrt{n} \times \sqrt{n}$, each being denoted as $A^{iq}$ and $B^{qj}$ respectively, where $ i,\ j \in [\sqrt{n}]$ and $ q \in [\frac{d}{\sqrt{n}}]$. Observe that 
$$C^{ij} = \sum_{q=1}^\frac{d}{\sqrt{n}}A^{iq}*B^{qj}, \text{where $*$ denotes matrix multiplication}$$

We can assume that each processor $p_{ij}$ contains one row and one column of $A$ and $B$ respectively . Note that an element of $A$ needs to be sent to at most $\sqn$ processors. Thus, $p_{ij}$ needs to send $d \sqn$ elements. In one round, it can send $n$ elements and thus, in $\frac{d}{\sqn}$ rounds, it can send all the elements to the destination processors. Now let us analyse how many elements of $A$ and $B$ are received by each processor. $A^{iq}$ and $B^{qj}$ are $\sqrt{n} \times \sqrt{n}$ matrices containing a total of $n$ elements. Therefore, $p_{ij}$ needs a total of $(\frac{d}{\sqn} \cdot n ) = d \sqn$ elements of $A$ and $B$ in order to compute $C^{ij}$. In one round, $p_{ij}$ can receive at most $O(n)$ elements and thus, in $\frac{d \sqn}{n} = \frac{d}{\sqn}$ rounds, it can receive all the elements that it requires for the computation of $C^{ij}$. 



\begin{theorem}
\label{thm:ndn}
 Multiplication of two rectangular matrices of size $n \times d$ and $d \times n$ respectively, with $n$ processors of $O(n)$ memory requires $\Theta(\frac{d}{\sqrt{n}})$ rounds, where $d \leq n$.
\end{theorem}

\section{Sparse Matrix Multiplication}
\label{sec:sparsemm}

In this section, we will discuss the case when $A$ and $B$ are $d$-sparse. Notice that in this case, the resulting matrix $C =A*B$ can be $d^2$-sparse. But in the same spirit as discussed in Gupta et al. \cite{gupta-2022-sparse}, in order to keep the model more meaningful, we assume that we are interested in only at most $d$ elements of each row and column of $C$. Therefore, both the input and output can be stored on $O(nd)$ memory. We analyse the case when we have $n$ processors, each with $O(d)$ memory. We assume that initially each processor holds one row and one column (i.e. $O(d)$ elements in total) of $A$ and $B$, and at the end each processor outputs $d$ elements of one of the rows of $C$. We assume that the position of these $d$ elements of $C$ that a processor needs to output is already known to the processor.



\paragraph*{A trivial algorithm:} Notice each element of $C$ is the sum of $d$ pairwise products of elements of $A$ and $B$. Thus, there are at most $ nd^2$ products that need to be calculated. In one round, a processor can compute any $d$ products of $C$ that are required to compute an element of $C$. It can do this by obtaining the $d$ elements of matrices $A$ and $B$, then performing $d$ multiplications locally. Thus, in one round of local computation, the $n$ processors can compute any $nd$ terms of $C$. Hence, to compute the $nd^2$ of $C$, we only require $O(d)$ communication rounds. Now, we give a better algorithm than this as follows.

We try to find a more efficient method to compute the terms of $C$. For that, we adapt the idea of Gupta et al. \cite{gupta-2022-sparse}. In particular, we use Lemma 4.2 of their paper. As we discussed earlier, in their paper, they translated the problem of matrix multiplication to triangle processing. So first, we translate their results according to the terminology we used in this. We call each product of type $a_{ij}b_{jk}$ a \textit{term}, where $a_{ij}$ and $b_{jk}$ are elements of $A$ and $B$ respectively. As we mentioned, there are at most $nd^2$ terms that needs to be computed inorder to compute $C$, Gupta et al. \cite{gupta-2022-sparse} computed them in a sequence (or layer) of rounds as follows. 

\begin{lemma}[\cite{gupta-2022-sparse}]
\label{lem:spaa22}
Let \(0 \le \varepsilon_1 < \varepsilon_2\), and assume \(d\) is sufficiently large. Let $T$ denote the remaining number of terms of $C$ that need to be computed such that 
$$|T| = O(d^{2-\varepsilon_1}n).$$
Then, we can compute these terms in $l$ sequential iterations where
$$l = O(d^{5\varepsilon_2-\varepsilon_1})$$
and after these iterations only $T'$ terms remain to be computed, such that
$$|T'| = O(d^{2-\varepsilon_2}n)$$

\end{lemma}

Let us try to explain the lemma in simple words. The lemma basically states that the multiplication of two $n \times n$, $d$-sparse matrices ($C = A * B$) that contain  $O(d^{2-\epsilon_1}n)$ \emph{terms} to be computed can be divided into two parts $C'$ and $C''$ such that

\begin{enumerate}
    \item $C' = A_1 \times B_1 + A_2 \times B_2 + \ldots + A_l \times B_l$ ($l$ is same as mentioned in the lemma) where 
    \begin{enumerate}[(i)]
        \item $A_i$  and $B_i$ are submatrices of $A$ and $B$ respectively.
        \item $A_i$  and $B_i$ consist of $(\frac{n}{d})$ submatrices $A_i^j$ and $B_i^j$, for $j \in [\frac{n}{d}]$, respectively. Each $A_i^j$ and $B_i^j$ is of size $d\times d$ such that $A_i \times B_i = \sum A_i^j *B_i^j$.
        
    \end{enumerate}
    
    \item $C'' = A' \times B'$ such that, it requires at most $O(d^{2-\epsilon_2}n)$ terms to compute $C''$

    \item $ C= C' +C''$
\end{enumerate}

Thus, the total time required to compute $C$ is the time to compute $C'$ and $C''$ one after another. The idea is to use square matrix multiplication as a black box algorithm to compute $C'$ (because computing $A_i$ and $B_i$ requires multiplication of many $d \times d$ square matrices) and use a brute force algorithm to compute $C''$. 

To compute $C'$, we can run $l$ sequential iterations, each for one $A_i \times B_i$. In one iteration, we can assign $d$ processors for the computation of $A_i^j *B_i^j$ for each $j$. These $(\frac{n}{d})$ computations will happen in parallel in one iteration. In this paper, we improve the bound on the number of iterations ($l$) required to go from $T$ to $T'$. Notice that, in the MPC model, the time required for one iteration is $O(\sqrt{d})$. Because this involves $\frac{n}{d}$ instance matrix multiplications of $d \times d$ matrices happening in parallel, each with the help of $d$ computers. Thus, if we plug $n=d$ in \cref{thm:ndn}, we get an $O(\sqrt{d})$ rounds algorithm for this task. Now we improve the bounds on the number of iterations required to bring down the number of terms from $O(d^{2-\epsilon_1})$ to $O(d^{2-\epsilon_2})$ mentioned in \cref{lem:spaa22}



\subsection{Tight Bound for number of Iterations}


\begin{lemma}
\label{lem:sparsemm-improved}
 Number of rounds required to bring down the number of uncomputed terms from $|T|\;=\; O(d^{2- \varepsilon_{1}}n)$ to $ |T'|\;=\; O(d^{2 - \varepsilon_{2}}n)$ is $l \;=\; O\!\bigl(d^{4\,\varepsilon_{2}}\bigr)$, where \(0 \le \varepsilon_1 < \varepsilon_2\).
\end{lemma}

\begin{proof}
According to \cref{lem:spaa22} it will require $O(d^{5\varepsilon_2-\varepsilon_1})$ iteration to go from $T$ to $T'$. Let us see what happens when we introduce an intermediate step in the reduction of terms from $T$ to $T'$, i.e. instead of directly going from $O(d^{2 - \varepsilon_{1}}n)$ to $O(d^{2 - \varepsilon_{2}}n)$. We first go from $O(d^{2 - \varepsilon_{1}}n)$ to $O(d^{2 - \varepsilon_{3}}n)$ and then go from $O(d^{2 - \varepsilon_{3}}n)$ to $O(d^{2 - \varepsilon_{2}}n)$ for some $\varepsilon_1< \varepsilon_3< \varepsilon_2$, using \cref{lem:spaa22}. And calculate the number of iterations consumed in this process. Let $\varepsilon_3 \;=\;\frac{\varepsilon_1 + \varepsilon_2}{2}$ such that
\[
|T|\;=\; O(d^{2 - \varepsilon_{1}}n)
\quad\text{,}\quad 
|T''|\;=\; O(d^{2 - \varepsilon_{3}}n)
\quad\text{, and }\quad
|T'|\;=\; O(d^{2 - \varepsilon_{2}}n).
\]

Let $l_1$ denote the number of iterations required to reduce the number of uncomputed terms from $T$ to $T''$, and similarly, $l_2$ denote the number of iterations required to reduce the number of uncomputed terms from $T''$ to $T'$. By \cref{lem:spaa22}, we have $l_1 = O(d^{5\varepsilon_3-\varepsilon_1})$ and $l_2 = O(d^{5\varepsilon_2-\varepsilon_3})$. Thus, the total number of iterations required to reduce the number of uncomputed terms from $T$ to $T'$ is
$l = l_1 + l_2 = O(d^{5\varepsilon_3-\varepsilon_1}) + O(d^{5\varepsilon_2-\varepsilon_3}) = O(d^{5\varepsilon_2-\varepsilon_3})$. Since $5\varepsilon_2-\varepsilon_3 > {5\varepsilon_3-\varepsilon_1}$, we get a strictly better bound than presented in \cref{lem:spaa22}.

Now the idea is that we recursively insert intermediate steps in this process i.e. insert $\varepsilon_4$ between $\varepsilon_3$ and $\varepsilon_2$ (such that $\varepsilon_4$ is the arithmetic mean of $\varepsilon_3$ and $\varepsilon_2$) and similarly $\varepsilon_5$ between $\varepsilon_4$ and $\varepsilon_2$ and do this $x$ number of times ( where $x$ is big constant), then, we can see that number of iterations required to reduce the number of \emph{terms} step by step would be $l = l_1 + l_2 + .. + l_{x+1}$, where $l_1$ is the number of iterations to go from $\varepsilon_1$ to $\varepsilon_3$, $l_2$ from $\varepsilon_3$ to $\varepsilon_4$, .., $l_{x+1}$ from $\varepsilon_{x+2}$ to $\varepsilon_2$. Observe that the order of $l$ is dominated by the term $l_{x+1}$. Thus, $l = O(d^{5\varepsilon_2-\varepsilon_{x+2}}) $. Therefore, by inserting enough intermediate steps (bounded by some constant) we can make  $\varepsilon_{x+2}$, infinitesimally $\varepsilon_{x+2}$, infinitesimally close to $\varepsilon_2$. Hence $l$ will converge to $$O(d^{5\varepsilon_2-\varepsilon_{x+2}}) \approx O(d^{4\varepsilon_2})$$
\end{proof}

\paragraph{Final Computation of terms:} As we mentioned, the multiplication of two $d$-sparse matrices can be done in two phases. In one phase, we iteratively use square matrix multiplication, and in the other phase, we use the brute-force algorithm. Our task is to compute $O(d^2n)$ terms. We will first reduce the number of uncomputed terms from $O(nd^2)$ to $O(nd^{2-\varepsilon})$ ($\varepsilon \geq 0$) by iteratively using square matrix multiplication, and then compute the rest of the terms using trivial computation. By \cref{lem:sparsemm-improved}, the number of iterations to reduce the number of uncomputed terms from $O(nd^2)$ to $O(nd^{2-\varepsilon})$ is  $O(d^{4\varepsilon})$, where each iteration takes $O(\sqrt{d})$ communication rounds. Thus the total time for phase one is $O(d^{4 \varepsilon+ \frac{1}{2}})$. After phase one, we are left with $O(nd^{2-\varepsilon})$ terms that we compute using a trivial brute force algorithm. We know that in one round, one processor can compute $O(d)$ terms (because it has $O(d)$ memory), therefore total rounds required to compute $O(nd^{2-\varepsilon})$ temrs will be $O(\frac{nd^{2-\varepsilon}}{nd}) = O(d^{1-\varepsilon})$. There is a tradeoff between the round complexity of these two phases, i.e. if we spend less time in phase one, it will take more time (polynomially) in phase two. Therefore, in order to optimise the total number of communication rounds needed, which is $O(d^{4\varepsilon+1/2})+O(d^{1-\varepsilon})$, should be minimised. Which happen for $\varepsilon=0.1$, taking $O(d^{0.9})$ communication rounds. Thus, we prove the following theorem.

\begin{theorem}\label{thm:sparse-constant-rounds}
Multiplication of two $d$-sparse matrices in the MPC model with \(n\) processors, each with \(\Theta(d)\) memory, can be performed in \(O(d^{0.9})\) rounds.
\end{theorem}

\bibliography{main}
\end{document}